\documentclass{article}
\usepackage{a4wide}
\usepackage{amssymb,amsfonts}

\def\H{\mathop{\hbox{\bf H}}}
\def\qed{$\square$}

\let\phi\varphi
\newtheorem{theorem}{Theorem}
\newtheorem{claim}[theorem]{Claim}
\newtheorem{lemma}[theorem]{Lemma}
\newtheorem{DDDefinition}[theorem]{Definition}
\def\enddefinition{\end{DDDefinition}\egroup\medbreak}
\def\proof{\begin{demo}{Proof}}
\def\endproof{\enspace\hfill\qed\end{demo}\medbreak}
\makeatletter
\def\definition{\bgroup\def\@begintheorem##1##2{\trivlist
  \item[\hskip\labelsep{\bfseries ##1\ ##2}]}\begin{DDDefinition}}
\makeatother
\newenvironment{demo}[1]
{\par\medbreak\noindent{\bf #1\enskip}\rm\ignorespaces}{}

\begin{document}

\title{Gruppen secret sharing\\
{\normalsize or}\\
how to share several secrets if you must?}
\author{L\'aszl\'o Csirmaz\thanks{This research was partially
supported by the ``Lend\"ulet Program'' of the Hungarian Academy of
Sciences}}
\date{\small Central European University, Budapest\\ R\'enyi Institute,
Budapest}
\maketitle

\begin{abstract}
Each member of an $n$-person team has a secret, say a password. The $k$
out of $n$ {\em gruppen secret sharing} requires that any group of
$k$ members should be able to recover the secrets of the other $n-k$ members,
while any group of $k-1$ or less members should have no information on the
secret of other team member {\em even if other secrets leak
out}. We prove that when all secrets are chosen independently
and have size $s$, then each team member must have a share of size at
least $(n-k)s$, and we present a scheme which achieves this bound when $s$
is large enough. This
result shows a significant saving over $n$ independent applications of
Shamir's $k$ out of $n-1$ threshold schemes which assigns shares of 
size $(n-1)s$ to each team member independently of $k$.

We also show how to set up such a scheme without any trusted dealer, and how
the secrets can be recovered, possibly multiple times, without leaking 
information. We
also discuss how our scheme fits to the much-investigated multiple secret
sharing methods.


\end{abstract}

\section{Introduction}

A team has $n$ members, and each member has a secret,
say a password. As a safety caution, they want each secret to be
distributed among other group members so that it could be
recovered in the case any of them would forget it. Also, none of them trusts 
the
others, thus they want their secrets to be independent of the information
held by any group $k-1$ or less team members -- even if other secrets leak
out. This goal can be achieved by
distributing all secrets using Shamir's $k$ out of $n-1$ threshold secret
sharing method, see \cite{shamir}. Assuming that all secrets are $s$ bit
long, the total size of the information each team member must remember will 
be $n\cdot s$ bits: $s$ bits for each team member plus her own password.

The question is: can we do better if the secrets are distributed 
simultaneously?

This question comes under the name of {\em multiple secret sharing},
which has two distinct flavors:
\begin{enumerate}
\item Different secrets are to be recovered by different access structures, 
usually only
one of the secrets will ever be recovered; also known as multiple {\em
secret-sharing}. A typical question is how much the
information to be remembered by each member can be squeezed compared to the
independent applications of traditional secret sharing. Results in 
this direction can be found in \cite{blundo-masucci, msss, ferras, jackson}.
\item A group recovers multiple secrets, this is {\em
multiple-secret} sharing. In this case in order
to decrease private information, unconditional security is traded
for computational security. See, e.g., \cite{das-adhikari, lin-yeh, zhao-zhao}.
\end{enumerate}
In both cases, {\em verifiable} schemes can also considered, where
participants can check whether shares provided by others are genuine or
not, look at \cite{das-adhikari, zhao-zhao}.

Our problem, which we call {\em $k$ out of $n$ gruppen secret sharing} 
(see Section \ref{sec:defs}), 
belongs to the first flavor (which, in fact, is more general
than the second one), as each secret is to be recovered by a different
collection of team members. In Sections \ref{sec:bounds} and
\ref{sec:protocol}
we concentrate on
the typical secret sharing question, and determine the amount
of information each participant must receive by proving a lower bound, and
giving a matching construction.

Section \ref{sec:recovery} looks at how one of the secrets can be recovered.
This is an intricate issue in a multiple secret settings as --
understandably -- we do not want to compromise others secrets when 
recovering someone's (allegedly) forgot password. This requirement is easy
to overlook. Interestingly, the ``private recovery process'' our
construction from Section \ref{sec:bounds} suggests still leaks out some 
information. We propose a perfect
solution at the expense of increasing the round complexity of the protocol.

Secret sharing methods, just as our construction does, usually refer to a
trusted {\em dealer} who knows all the secrets, distributes the shares
privately, and disappears after doing her job keeping all secrets. The
homomorphic property of our construction suggests another setup performed by
the participants without any trusted third party. In Section
\ref{sec:beginning} we look at it in details, and conclude that it is, in
fact, secure, and does not leak out information.

\section{The ``gruppen secret sharing''}\label{sec:defs}

Multiple secret sharing schemes are a natural generation of single secret
sharing schemes, and been defined formally, among others, in
\cite{blundo-masucci}. Informally, in such a scheme we have a set $P$ of
{\em participants}, and $n$ {\em access structures} $\mathcal A_1$, $\dots$,
$\mathcal A_n$, that is, upward closed collections of subsets of the
participants. The {\em dealer} picks (or receives) an $n$-tuple of secrets
$\langle s_1,\dots, s_n\rangle$ from some finite domain with a given
distribution (typically secrets are independent and uniformly
distributed), and computes, using some randomness, the {\em shares} of the
participants.

\begin{definition}\label{def:sound-and-perfect}
The multiple secret sharing scheme $\mathcal S $ is {\em sound} if qualified
subsets can recover the secret: whenever $A\in \mathcal A_i$, then members
of $A$, using their private information only, can recover the secret $s_i$.

The scheme $\mathcal S$ is {\em perfect}, if $B\subseteq P$ is not enabled
to recover the secret $s_i$ (that is, $B\notin \mathcal A_i$), then members
in $B$, {\em even knowing all other secrets $s_j$ for $j\not= i$}, have no more information on
$s_i$ than that already conveyed by the legally known values.
\end{definition}

In particular, if the secrets are independently chosen, then the totality
of shares of $B \notin \mathcal A_i$ should give no information on $s_i$
whatsoever {\em even given all other secrets}. 

\bigskip
A {\em gruppen secret sharing} scheme is a special, nevertheless
interesting,
case of multiple secret sharing schemes. Each participant has a secret drawn 
uniformly and 
independently from some finite domain $S$, say a password. Each password
should be recoverable by any $k$ out of the remaining $n-1$ participants,
but no coalition of $k-1$ or less participants should know anything about 
the remaining secrets.

\begin{definition}\label{def:gruppen}
In a {\em $k$ out of $n$ gruppen secret sharing scheme} there are $n$
participants in $P$, participant $i\in P$ has a secret $s_i$ drawn 
uniformly and independently from some domain, and the access structure
$\mathcal A_i$ -- whose members should be able to recover $s_i$ -- consists
of all subsets of $P-\{i\}$ with at least $k$ members, that is the $k$ out
of $n-1$ threshold structure on $P-\{i\}$.
\end{definition}
For definiteness we assume that all secrets are $s$ bit long (random)
0--1 sequences where $s$ is large enough.

\bigskip
It is easy to construct a sound and perfect $k$ out of $n$ gruppen
secret sharing scheme. For each participant $i\in P$, the dealer distributes
the secret $s_i$ to members of $\mathcal A_i$ {\em independently} using
Shamir's $k$ out of $n-1$ threshold scheme. As this latter scheme
is also {\em perfect}, i.e., everyone gets a minimal
size share (which is $s$ bits), everyone receives $n-1$ shares of $s$ bit
each, next to his $s$ bit secret, which is a total $n\cdot s$ bits to
remember. Is there any way to do it better? The next
theorem answers this question.

\begin{theorem}\label{thm:main}
a) In a perfect and sound $k$ out of $n$ gruppen secret sharing scheme each
participant must receive a share of at least $(n-k)s$ bits. 

b) For $s$ large enough there is a perfect and sound $k$ out of $n$ gruppen 
secret sharing scheme
where every participant receives exactly $(n-k)s$ bit share.
\end{theorem}

We postpone the proof to Sections \ref{sec:bounds} and
\ref{sec:protocol}; here we illustrate the theorem for the case when $n=3$
and $k=2$. We have three participants whom we call Alice, Bob, and Cecil,
having secrets $a$, $b$, and $c$, respectively. The lower bound on the share
size is almost immediate. Bob has no information on Alice's and Cecil's
secret. When Alice joins Bob, the two of them have enough information to
determine both Alice's and Cecil's secret. This means $2s$ bits of
information which should come from Alice. Her secret is $s$ bit long, thus
her share must also be at least $s$ bit long to supply that much
information.

As for the construction, the dealer should tell them  
shares which have size equal to that of the secrets, so that
a) any two participant should be able to recover the secret of the third; 
and b) no one should have any information on the others' secrets.

To satisfy the first requirement, our first attempt is to give Alice, Bob,
and Cecil the shares $c$, $a$, and $b$, respectively. This way any two can
recover the secret of the third one, but these shares definitely contradict
the security requirement in Definition \ref{def:sound-and-perfect}.
So we ``hide'' these shares by xoring them with some value
possessed by others: let the three shares be $c\oplus b$, $a\oplus c$,
$b\oplus a$, respectively. Again, any two can recover the third's secret,
for example Alice knows $c\oplus b$, Bob knows $b$, thus they can recover $c$.
Unfortunately these shares also violate the security requirement. If Bob's
secret leaks out, or Alice simply guesses it right, then Alice alone could
recover Cecil's $c$. Also, if $b$ and $c$ are weak (and long) passwords,
then it is a simply routine to recover both $b$ and $c$ from $c\oplus b$.

A solution could be using interpolating polynomials \`a la Shamir. Let $r$
be a random polynomial which takes the secrets at its first three places:
$r(0)=a$, $r(1)=b$, and $r(2)=c$. Then let the shares be $r(3)$, $r(4)$, and
$r(5)$, respectively. Any pair of participants knows $r$ at four different
places, thus if $r$ has degree at most $3$, then they can recover the
polynomial $r$, thus
the third participant's secret as well. The security requirement also holds:
a single participant knows $r$'s value at two places. If one of the
two other secrets leak out, then it is $r$'s value at an additional place.
Given $r$'s value at (at most)
three places provides no information about what values $r$ can take at a
fourth place, and this is what was required.

As usual, the polynomial $r$ is over some finite field; the secrets are
random elements from this field. The above scheme will work when the field
has at least six distinct elements, thus we must have $s\ge 3$.

\section{You cannot do better than \dots}\label{sec:bounds}
In this section we show that the amount of share every participant in a
$k$ out of $n$ gruppen secret sharing scheme must have is at least
$s\cdot(n-k)$ bits, where every secret is an (independent, uniformly random)
$s$ bit long 0--1 word. This proves the first part of Theorem
\ref{thm:main}.

First we give an informal reasoning, then we make it precise using the
entropy method \cite{capocelli,entropy}.
Let $a_1,\dots,a_{k-1}$, and $b\in P$ be $k$ different participants. We want
to estimate the amount of private information $b$ must have. By assumption,
the totality of the private information the group $\{a_1,\dots,a_{k-1},b\}$
has determines uniquely the secrets of the remaining $n-k$ participants,
which amounts to $(n-k)\cdot s$ bits. By the security requirement,
whatever $\{a_1,\dots,a_{k-1}\}$ know should be independent of those secrets
{\em plus} the secret of $b$. Thus the additional $(n-k)\cdot s + s$ bits of
information must be supplied by $b$. He has $s$ bits of secret, thus must
have a share of size at least $(n-k)\cdot s$ bits.

\bigskip
\newcommand\secret{\mathrm{secret}}
\newcommand\share{\mathrm{share}}
The above reasoning can be made precise using the so-called entropy method
as described in, e.g., \cite{capocelli} or \cite{entropy}. First of all, we
consider the secrets and shares as random variables.
The {\em size} of the value of a random variable $\xi$ is its Shannon entropy
 $\H(\xi)$, which
is (roughly) the number of necessary (independent) bits to define the value
of $\xi$ uniquely. Our assumption was that the secrets are independent
$s$ bit long 0--1 sequences, thus $\H(\xi\eta)=2s$, where $\xi$ and $\eta$
are the secret values of two participants.

For any collection $\{\xi_i: i\in I\}$ of random variables define the
real-valued function
$$
    f(I) = \H(\{\xi_i : i\in I\} )
$$
where the entropy is taken for the joint distribution of all indicated
variables. For example, if $a$ and $b$ are (indices of) two participant's
secrets, then $f(\{a\})=s$, and $f(\{a,b\})=2s$ as we have seen above.
Similarly, if $j$ is (an index of) a share, then $f(\{j\})$ is
the {\em size} of that share.

The function $f$ is defined on all subsets of some finite set, and satisfies
certain linear inequalities which follow from the so-called Shannon
inequalities for the entropy function $\H$. The following claim collects
those properties which will be used to prove the theorem. As usual, we write
$f(XY)$ instead of $f(X\cup Y)$, and $f(x)$ and $f(xX)$ instead of
$f(\{x\})$ and $f(\{x\}\cup X)$.
\begin{claim}[\rm See \cite{entropy}]\label{claim:entropy-method}
For any subsets $X$ and $Y$
\begin{enumerate}
\item $f(X)\ge 0$ (positivity), 
\item $f(X)\le f(Y)$ if $X\subseteq Y$ (monotonicity),
\item $f(X)+f(Y)\ge f(XY)$ (additivity),
\item $f(XY)=f(X)$ if (the variables in) $X$ determines the values of (the variables
in) $Y$;
\item $f(XY)=f(X)+f(Y)$ if $X$ and $Y$ are statistically independent.
\end{enumerate}
\end{claim}

The entropy method can be rephrased in a few words as follows. Let $j$ be
the (index) of any share. Suppose for {\it any function $f$} satisfying 
properties enlisted in Claim \ref{claim:entropy-method} there are (indices)
$a_1$, $\dots$, $a_\ell$ of secrets such that
$$
    f(j) \ge f(a_1)+\cdots f(a_\ell).
$$  
The the size of share $j$ must be at least $\ell$ times the size of the
secrets.

\begin{lemma}\label{lemma:auxsize}
Suppose $G$ is a group of participant with $k-1$ members, and $a$, $\bar
b=\langle b_1,\dots,b_{n-k}\rangle$ are the (indices of the)
secrets of participants not in $G$ and (the indices of) their shares are 
$j$ and $\bar \jmath
=\langle j_1,\dots,j_{n-k}\rangle$, respectively. Then
$$
    f(a)+f(j) \ge f(a \bar b).
$$
\end{lemma}
\begin{proof}
Let us denote the total data (secret plus share) held by $G$ by
$G$ as well. By assumption, $G$ together with $a$ and $j$ determines
all the secrets $b_i$, that is $f(ajG) = f(a\bar bjG)$. Also, $G$ should 
have no information on the secrets $a$ and $b_i$ or on their combinations, 
thus $f(a\bar bG) = f(a\bar b) + f(G)$. Using these, the additivity and 
monotonicity property of $f$, we have
$$
    f(a)+f(j)+f(G) \ge f(ajG) = f(a\bar bjG) \ge f(a\bar bG) = 
          f(a\bar b)+ f(G) .
$$
Comparing the first and last tag gives the claim of the Lemma.  
\end{proof}

From this lemma we can easily deduct the required lower bound on the size of the share
each participant receives.

\begin{proof} (of first part of Theorem \ref{thm:main})
Use notations from Lemma \ref{lemma:auxsize}, in particular let $a$ and $j$
respectively be (indices of) the secret and the share of participant $a$. All
secrets have the same size, thus
$$
    f(a)=f(b_1)=\cdots = f(b_{n-k}).
$$
By assumption the secrets are totally independent, which means
$$
   f(a\bar b) = f(ab_1\dots b_{n-k})=f(a)+f(b_1)+\cdots+f(b_{n-k})
      =(n-k+1)f(a).
$$
From Lemma \ref{lemma:auxsize} we know that $f(j)\ge f(a\bar b) - f(a) =
(n-k)f(a)$, which proves part a) of Theorem \ref{thm:main}.
\end{proof}

\section{An (optimal) protocol}\label{sec:protocol}

Our $k$ out of $n$ gruppen secret sharing scheme, whose complexity matches
the bound given in Section \ref{sec:bounds}, is a straightforward
generalization of the one sketched in Section \ref{sec:defs}. Let $\mathbb
F$ be a finite field with more than $n(n-k+1)$ elements. Secrets will be 
chosen uniformly and
independently from $\mathbb F$, which means that if secrets are $s$ bit long
0--1 sequences, then $\mathbb F$ can be chosen to be the field of
characteristic $2$ on $2^s$
elements. To give a scheme means to describe how the dealer computes
(determines) the shares given the randomly and uniformly chosen secrets; or,
equivalently, how the dealer can distribute the shares {\em and the
secrets} simultaneously as long as the secrets
come from the appropriate distribution. We will choose this latter approach,
and hint how to modify the scheme when the secrets are given in advance.

Let $p_i$ for $1\le i \le n$ denote the participants. The dealer chooses
different field elements $x_{i,j}$ for $1\le i\le n$ and $0\le j\le n-k$,
and picks a random polynomial $r(x)$ over $\mathbb F$ of degree less than
$k(n-k+1)$.

The {\em secret} of participant $p_i$ will the the value of $r$ at
$x_{i,0}$. (When given the secrets in advance, the same distribution can be
achieved by simply choosing $r$ randomly from among those polynomials which
give the secret of $p_i$ at $r(x_{i,0})$.) As for the share, the dealer
gives participant $p_i$ all field elements $r(x_{i,1})$ up to
$r(x_{i,n-k})$. Observe that secrets are uniform random elements from
the field, thus the ``size'' (entropy) of every secret is the same, namely
$\log_2(|\mathbb F|)$. Similarly, all participants receive $(n-k)$ field
elements as share, therefore the size of the share is exactly $(n-k)$ times
that of the secret.

We claim that any $k$ participants can determine the secret value
of the remaining $n-k$ participants. This is clear, as the $k$ participants
know the value of $r$ at $k(n-k+1)$ different places, while $r$ has
smaller degree, thus they can determine $r$, and its value at $x_{p,0}$ for
any participant $p$.

Next, we claim that the total information of $k-1$ participants is
statistically independent of the secrets of the other $n-k+1$ participants.
This is true as $r$ is a random polynomial of degree below $k(n-k+1)$, and
$k-1$ participants know the value of this polynomial at $(k-1)(n-k+1)$
places, thus the polynomial can take all the possibilities with equal
probability at any $n-k+1$ predetermined places -- in particular at
$x_{p,0}$ where $p$ runs over the missing $n-k+1$ participants.
Consequently, all private information of $k-1$ participants, {\em plus} the
secret of all but one remaining participants is statistically independent of
the secret of the last participant. This is exactly the security requirement
which proves the second part of Theorem \ref{thm:main}.

\section{Secret recovery}\label{sec:recovery}
The method outlined in the previous Section to recover the secret of $p\in
P$ was that $k$ participants, using their private values, recover the
polynomial $r$, and then compute $r$'s value at $x_{p,0}$. This
recovery process has the drawback that once $r$ is known, all secrets are
revealed, not only the secret of $p$.
How can we achieve that they
recover the value of $r$ at $x_{p,0}$ only and not the whole polynomial
$r$? Let $B \subseteq \{1,\dots,n\}$ be the subset of size $k$ which wants
to recover the secret of $p\notin B$. As the values $x_{i,j}$ are publicly
known, everyone can compute the constants $\lambda_{i,j}\in\mathbb F$
using, e.g., the Lagrange interpolation formula such that
$$
    r(x_{p,0})=\sum_{i\in B} \sum_{j=0}^{n-k} \lambda_{i,j}r(x_{i,j})
$$
independently what the values $r(x_{i,j})$ are.
Consequently to recover $p$'s secret, participant $i\in B$ should only
compute the sum
\begin{equation}\label{eq:sum}
   t_{i,p} =  \sum_{j=0}^{n-k}\lambda_{i,j}r(x_{i,j})
\end{equation}
and send it privately to $p$, rather than revealing all the $r(x_{i,j})$ 
values. $p$ receives the $k$ values (\ref{eq:sum}) from participants in
$B$, he simply adds them up to recover his secret.

\medskip
Unfortunately this process leaks out information, and cannot be repeated
indefinitely. To see why this is the case, we go back to the 2 out of 3
gruppen secret sharing scheme as discussed in Section \ref{sec:defs}. Alice,
Bob, and Cecil have secrets $r(0)$, $r(1)$, and $r(2)$, and have shares
$r(3)$, $r(4)$, and $r(5)$, respectively. Now Alice announces that she lost
her secret. Lagrange says that
$$
   r(0)=\frac{10}3r(1) + \frac53r(4) - \frac{10}3f(2)-\frac23r(5),
$$
thus Bob sends Alice the value
$$
   t_b = \frac{10}3r(1) + \frac53r(4),
$$
and Cecil sends
$$
   t_c =  - \frac{10}3f(2)-\frac23r(5).
$$
The Alice could recover her secret as $t_b+t_c$. However, Alice could be
cheating, as she still have her share $r(3)$. Again by Lagrange
$$
   r(3) = -\frac16r(1) + \frac23r(4) + \frac23 r(2)- \frac16r(5).
$$
Alice can eliminate $r(4)$ and $r(5)$ using the values she received from Bob
and Cecil, thus she knows
$$
   \frac23\big(r(3) - \frac25t_b - \frac14t_c\big) = -r(1) +  r(2)
$$
a nontrivial combination of Bob's and Cecil's secrets. Thus if any of those
two secrets leak out, or Alice could successfully guess it, she'll know the
other secret immediately.

\medskip 
Switching to linear algebra from polynomial interpolation, a random
polynomial of degree less than $k(n-k+1)$ can be considered as a random
vector if the $k(n-k+1)$-dimensional space. Knowing the value at a certain
place amounts to know a (fixed) linear combination of the coefficients of
the random vector. Initially every participant knows $(n-k+1)$ such linear
combinations. Thus the codimension of $k-1$ participants private information
is $n-k+1$, it is just the linear space where the remaining $n-k+1$
participants have their secrets.

During the recovery procedure $p$ receives $k$ further linear combinations
(\ref{eq:sum}). As these add up to his secret, the number of new linear 
combinations he knows is $(k-1)$ more. Thus if this process $p$ is joined 
by $k-2$ other participants, the codimension of their information is
$n-2(k-1)$, thus must leak some information about the other's secrets.

\medskip
A possible remedy is to keep track the codimension of the total information
of any $k' \le k$ participants, and let run the recovery process until it is
large enough. Also, if $p$ announces that he lost his secret, then $p$
should be excluded in any further recovery stage.

Another remedy is to use freshly generated random values which hide the
exact values of the sums (\ref{eq:sum}) from $p$. This solution has the
drawback that it increases the communication overhead, requires some further
(trivial) communication. However it has further advantages:
\begin{itemize}
\item anyone's secret can be recovered arbitrary number of times without
affecting the security level;
\item not only the secret, but also the shares can be recovered without
significant increase in the communication, thus recovering the ``full
state'' after a break down.
\end{itemize}
As before, let $B$ be the $k$-element set of participants who want to
recover $p$'s secret. Each $i\in B$ generates $k$ random and independent
elements from $\mathbb F$, say $r_{i,j}$, $j\in B$. Then $i$ sends
$r_{i,j}$ to $j$. After this step $i$ will know all $r_{i,j}$ (as he 
generated those numbers), and $r_{j,i}$ (as he received them from the
others). After this $i$ sends $p$ the obfuscated element
\begin{equation}\label{eq:randomsum}
   t'_{i,p} = t_{i,p} + \sum_{j\in B}(r_{i,j}-r_{j,i}) .
\end{equation}
After receiving all sums in (\ref{eq:randomsum}), $p$ simply adds them up
and recovers his secret.

Rather than interpolating the polynomial $r$ at the place $x_{p,0}$ only,
participants in $B$ can interpolate $r$ at every $x_{p,j}$ and compute the
sums similar to (\ref{eq:sum}). In the
obfuscating step everyone generates $k(n-k+1)$ random elements ($k$ for each
$j$) independently, and then sends the $(n-k+1)$ obfuscated interpolation
sums to $p$, who can recover his secret plus all the shares. This way no
private information is leaked out, and the whole process can be repeated
indefinitely.

\section{How to distribute the shares}\label{sec:beginning}

Any secret sharing scheme relies on a trusted {\em dealer} to set up the
scheme, who collects the secrets from the participants, generates the
shares, and tells the every participant her share privately, and then
disappears without leaking out any information. Such a trusted entity
is quite hard to find, protocols not relying on trusted party are
preferable to ones which use one. Fortunately, the scheme described in
Section \ref{sec:protocol} has the {\em homomorphic property}: if a scheme
distributes secrets $s_i$ and with shares $h_i$ for $i\in P$, another scheme
distributes secrets $s'_i$ and has shares $h'_i$, then for the secrets
$s_i+s'_i$ the shares $h_i+h'_i$ are correct ones, and have the appropriate
distribution. Here the addition is the addition in the field $\mathbb F$.

Using this homomorphic property, a gruppen secret sharing scheme can be
set up by the participants as follows. Suppose participant $i\in B$ has the
secret $s_i$. He computes, as a dealer, the shares of the $k$ out of $n$
gruppen secret sharing scheme as described in Section \ref{sec:protocol}, 
where all secrets are zero, except his own,
which is $s_i$. He generates the shares $h_{i,j}$ for $j\in P$, and then
sends $h_{i,j}$ to participant $j$.

Each participant receives shares from everyone else (including himself), and
his share in the final scheme is just the sum of all shares received. As a
consequence of the homomorphic property, in this way the participants
achieved a correct scheme which distributes their secrets. In an ideal
scheme the participant $p\in P$ receives only his share from the dealer; now
every participant receives an $(n-k)$-dimensional vector of $\mathbb F$ from
the others such that his share is the sum of the vectors received. Thus it
might happen that a coalition of $k-1$ participants could extract extra
information from the values they received. We claim that this is not the
case, this set-up is, in fact, $k-1$-secure.

Let us fix a coalition $B$ of $k-1$ participants, and look at what they
receive from $p\notin B$. $p$ generates a random polynomial of degree ${}<
k(n-k+1)$ which takes zero at $x_{i,0}$, $i\not= p$, and $p$'s secret at
$x_{p,0}$. To make the polynomial random, its value should be given {\em
randomly and independently} at further $k(n-k+1)-n = (k-1)(n-k)$ places. Now
participants in the coalition $B$ receive the polynomial's value at exactly
$(k-1)(n-k)$ places, thus the random polynomial can be set up by choosing its
value at these places randomly and independently. This also means that
members of $B$ can extract no information from the values they receive from
$p$.


\begin{thebibliography}{9}

\bibitem{blundo-masucci}
C.~Blundo, B.~Masucci,
\newblock {\em Randomness in multi-secret sharing schemes},
\newblock Journal of Universal Computer Science, Vol 5(7) (1999) pp.
367--389

\bibitem{msss}
C.~Bludo, A.~De~Santis, G.~Di~Crescenzo, A.~G.~Gaggia, U.~Vaccaro,
\newblock Multi-secret sharing schemes,
\newblock in {\em Advances in Cryptology -- CRYPTO '94}, Y. G. Desmedt, ed.,
Lecture Notes in Computer Science   839  (1994), pp. 150-163

\bibitem{capocelli}
R.~M.~Capocelli, A.~De~Santis, L.~Gargano, U.~Vaccaro:
\newblock On the size of shares of secret sharing schemes,
\newblock {\em J. Cryptology}, vol 6(3) (1993), pp. 157--168

\bibitem{entropy}
L.~Csirmaz:
\newblock Secret sharing schemes on graphs,
\newblock {\em Studia Sci. Math. Hungar.}, vol 44(2007) pp. 297--306
\newblock -- available as IACR preprint
{\tt http://\allowbreak eprint.iacr.org/\allowbreak 2005/059}

\bibitem{das-adhikari}
A.~Das, A.Adhikari,
\newblock An efficient multi-use multi-secret sharing scheme based
               on hash function,
\newblock {\em Applied Mathematics Letters}, 23 (2010) pp. 993--996

\bibitem{ferras}
O.~Ferras, I.~Gracia, S.~Martin, C.~Padro,
\newblock {\em Linear threshold multisecret sharing schemes},
\newblock ICITS 2009, pp. 110--126

\bibitem{jackson}
W.-A.~Jackson, K.~M.~Martin and C.~M.~O'Keefe, 
\newblock {\em Multisecret threshold schemes}, 
\newblock in {\em Advances in Cryptology -- CRYPTO '93}, D. R. Stinson, ed.,
Lecture Notes in Computer Science 773 (1994), pp. 126--135.


\bibitem{lin-yeh}
Han-Yo Line, Yi-Shiung Yeh,
\newblock Dynamic multi-secret sharing scheme
\newblock Int. J. Contemp. Math. Sciences, Vol 3(1) (2008), pp. 37--42

\bibitem{shamir}
A.~Shamir:
\newblock How to share a secret,
\newblock {\em Commun. of the ACM}, vol 22 (1979) pp. 612--613

\bibitem{zhao-zhao}
J.~Zhao, J.~Zhang, R.~Zhao,
\newblock A practical verifiable multi-secret sharing scheme
\newblock Computer Standards And Interfaces 29 (2007) pp. 138--141

\end{thebibliography}
\end{document}